\documentclass[a4paper,fleqn]{cas-dc}

\usepackage{amsthm}   
\usepackage[authoryear]{natbib}

\usepackage{tikz}
\usetikzlibrary{arrows.meta,positioning}

\ExplSyntaxOn
\bool_gset_true:N \g_stm_nologo_bool
\ExplSyntaxOff

\newtheorem{hypothesis}{Hypothesis}
\newtheorem{theorem}{Theorem}
\newtheorem{corollary}{Corollary}
\newtheorem{definition}{Definition}
\theoremstyle{remark}
\newtheorem{remark}{Remark}
\theoremstyle{plain}

\begin{document}


\shorttitle{Quantum Cryptographic Exposure under HNDL Threat}
\shortauthors{Rufino, Marcelino, and Garcia}


\title[mode=title]{A Formal Basis for Quantum Cryptographic
  Exposure Measurement under HNDL Threat}

\author[1]{Matheus Rufino\corref{cor1}}[orcid=0000-0003-4871-5120]
\ead{math.rufi@gmail.com}

\author[1]{Rafael Duarte Marcelino}[orcid=0009-0008-2081-6812]

\author[1,2]{Julio Smanioto Garcia}[orcid=0009-0006-7890-5821]

\address[1]{GWK Security, Campinas, SP, Brazil}
\address[2]{Instituto de F\'{i}sica Gleb Wataghin, Universidade Estadual de
  Campinas (UNICAMP), Campinas, SP, Brazil}

\cortext[cor1]{Corresponding author.}

\begin{abstract}
An adversary copies your encrypted traffic today and waits for a
quantum computer to decrypt it later.
How exposed are you?

We show that the functional form of the answer is not merely a
calibration choice --- it is structurally justified by three
assumptions about adversarial production and value-decay dynamics.
Under those assumptions, the HNDL compromise probability factorises
into a temporal hazard, a multiplicative cryptographic-vulnerability
and operational-exposure term, and a saturation denominator governed
by the defense--attack intensity ratio;
the marginal sensitivity to each dimension is endogenous to the
organisation's position in the vulnerability--exposure plane,
not a fixed global constant.
Additive scoring frameworks cannot reproduce this structure because
the interaction between cryptographic vulnerability and operational
exposure is absent by construction, regardless of calibration.
The resulting framework provides a structurally grounded basis for
operational HNDL exposure prioritisation under partial observability.
\end{abstract}

\begin{keywords}
HNDL attack \sep post-quantum cryptography \sep cryptographic exposure
index \sep contest success function \sep multiplicative risk model \sep
composite indicator
\end{keywords}

\maketitle

\section{Introduction}

The HNDL threat has a simple structure.
An adversary copies your encrypted traffic today and stores it.
When a cryptographically relevant quantum computer (CRQC) becomes
available, the stored traffic becomes readable.
This is the Harvest-Now-Decrypt-Later (HNDL) attack~\citep{joseph2022,mosca2018}.

The relevant question is not whether your cipher is breakable today.
It is whether it will be broken before the data loses its value.
NIST has finalized its first post-quantum standards~\citep{nist8547,nistfips2024}.
CISA, NSA, and NIST have published migration frameworks~\citep{cisa2023};
the European Union Agency for Cybersecurity has published a concurrent PQC
readiness survey~\citep{enisa2021}.
Organizations face a concrete decision: how urgently must they migrate?
The answer requires a score that measures HNDL exposure, not merely
cryptographic age.

Operational HNDL assessment is also constrained by partial
observability: external attack-surface signals, incomplete
cryptographic inventories, and sector-specific data lifetimes must be
mapped into a coherent exposure model.
The question addressed here is the structural form such a model must
have before any implementation-specific calibration is attempted
(Figure~\ref{fig:pipeline}).

\begin{figure}[t]
\centering
\begin{tikzpicture}[
  box/.style={draw, rounded corners=3pt, minimum width=4.8cm,
               minimum height=0.72cm, align=center, font=\small},
  side/.style={draw, rounded corners=3pt, minimum width=1.9cm,
                minimum height=0.72cm, align=center, font=\small},
  arr/.style={-{Stealth[length=4pt]}, thick},
  node distance=0.52cm and 0.55cm
]
  \node[box] (sig) {Observable \& declarative signals};
  \node[box, below=of sig] (map) {Structural mapping to $V$,\;$E$};
  \node[box, below=of map] (csf) {Proportional-hazards contest};
  \node[box, below=of csf] (ieq) {HNDL exposure score $\in[0,100]$};
  \node[side, right=of csf]  (h)  {Temporal\\hazard $H$};
  \draw[arr] (sig) -- (map);
  \draw[arr] (map) -- (csf);
  \draw[arr] (csf) -- (ieq);
  \draw[arr] (h)   -- (csf);
\end{tikzpicture}
\caption{Structural organisation of the HNDL exposure assessment
  problem.
  Observable and declarative signals are mapped to the structural
  variables $V$ (vulnerability fraction) and $E$ (operational
  exposure); the proportional-hazards contest
  (Hypotheses~\ref{hyp:poisson}--\ref{hyp:composition}) and
  temporal hazard $H$ (Definition~1) produce the exposure score of
  Eq.~\eqref{eq:IEQ}.
  The contribution of this paper is the structural form of the lower
  two stages; the signal-to-$(V,E)$ mapping is
  implementation-specific.}
\label{fig:pipeline}
\end{figure}
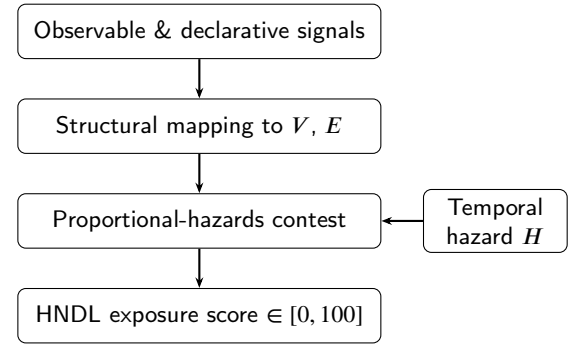

The natural instinct is to add things up.
Weight the vulnerability, weight the exposure, maybe add a governance
factor.
Security scoring approaches of this type are widespread~\citep{gordon2002,oecd2008}.
The OECD composite-indicator handbook~\citep{oecd2008} calls the
underlying assumption \emph{full compensability}: a deficit in one
dimension is perfectly offset by a surplus in another.
Under this assumption, vulnerability and exposure are substitutes.
Under the HNDL model, they are complements.
An organization with high vulnerability but zero external exposure
cannot be harvested.
An organization with zero quantum-vulnerable cryptography has nothing
to harvest.
Both dimensions must be present simultaneously for the attack to work.
No choice of weights in an additive score recovers this interaction.

This paper characterizes the functional form an HNDL exposure score
must take.
Three structural axioms support the multiplicative form.
Section~\ref{sec:setup} introduces the four model quantities.
Section~\ref{sec:hypotheses} states the three structural hypotheses.
Theorem~\ref{thm:main} (Section~\ref{sec:main}) derives the
model-implied form and the operational index.
Corollary~\ref{cor:degenerate} (Section~\ref{sec:consequences}) shows
that no additive score can preserve the required $V \times E$
interaction structure.
Section~\ref{sec:empirical} reports an internal specification diagnostic
and Section~\ref{sec:limitations} states the main limitations.
The contribution is structural rather than implementation-specific:
we characterize the admissible functional form of an HNDL exposure
score and show why additive or ordinal composites cannot preserve
the required interaction structure.
We refer to this score as the \emph{Quantum Exposure Index}
(IEQ).\footnote{IEQ is the acronym of the Portuguese
  \emph{{\'{I}}ndice de Exposi\c{c}\~ao Qu\^{a}ntica}, the operational
  name of the deployed implementation surveyed in
  Section~\ref{sec:empirical}.}

\section{Setup}
\label{sec:setup}

Four quantities characterize an organization under the HNDL threat.

$V \in (0,1]$ is the \emph{quantum vulnerability fraction}: the share of
its cryptographic attack surface using algorithms breakable by
Shor's algorithm --- RSA, ECDH, ECDSA, and DSA: algorithms based
on integer factorization or discrete logarithm
problems~\citep{nist8547,etsi_qsc001}.
A value of $V = 1$ means every cryptographic endpoint is
quantum-vulnerable.

$E \in (0,1]$ is the \emph{operational exposure}: how accessible
that surface is to an external adversary.
High $E$ means the encrypted traffic is observable and storable.
Low $E$ means the adversary cannot reach it.

$T_D > 0$ is the \emph{adversarial shelf life}: how long the captured
ciphertext retains strategic value after harvest.
A trade secret may be sensitive for decades~\citep{joseph2022}.
A stock price tip expires in days.

$\mu > 0$ denotes the effective rate at which already-harvested
ciphertext loses strategic exploitability --- through data-value decay,
lifecycle controls, key and material rotation, cryptographic remediation,
and other mechanisms that reduce the usefulness of the harvested
material~\citep{barker2020,nist8547}.
PQC migration is one component of this process and, separately, limits
the rate of future harvests.

Let $T_Q$ be the (random) time at which a CRQC becomes available.
The quantity of interest is the probability that the CRQC arrives
before the harvested data loses its strategic value \emph{and}
exploitation precedes remediation --- the HNDL compromise probability,
denoted $P_{\mathrm{HNDL}}$. If the CRQC does not arrive within
the data's adversarial horizon, the harvested ciphertext remains safe.

\section{Hypotheses}
\label{sec:hypotheses}

\begin{hypothesis}[Competing exponential processes]
\label{hyp:poisson}
  Think of the attacker and defender as running two races.
  The attacker tries to exploit the harvested data before the defender
  renders captured ciphertext strategically valueless through data
  expiry, operational rekeying, and lifecycle management.
  We model both races as constant-rate processes:
  \begin{equation}
    \lambda_A = \lambda_0\, V^a E^b, \qquad \lambda_D = \mu,
    \label{eq:rates}
  \end{equation}
  with $a, b > 0$ and $\lambda_0 > 0$ the baseline adversarial
  intensity.
  The multiplicative structure follows directly from the
  \emph{intersection principle}: a cryptographic compromise requires
  the \emph{simultaneous} occurrence of two approximately independent
  conditions --- the adversary can operationally reach the asset ($E$)
  and that asset uses quantum-vulnerable cryptography ($V$).
  Under the approximate cross-sectional independence of
  Hypothesis~\ref{hyp:independence},
  \[
    P(\text{reach} \cap \text{vulnerable})
      \;\approx\; E^b \times V^a,
  \]
  where the exponents encode structural properties of the
  infrastructure: $a \geq 1$ reflects concentration of critical assets
  (HSMs, CAs, TLS gateways) in the vulnerable subset, and $0 < b < 1$
  reflects attack-surface saturation yielding diminishing marginal
  returns to additional exposure vectors.
  Axiomatically, this family is further supported by three structural
  axioms established by \citet{skaperdas1996} for contest success
  functions:
  \begin{enumerate}
    \item[(A1)] \textit{Anonymity}: the attack rate depends only on
      $V$ and $E$, not on who the organization is.
    \item[(A2)] \textit{Independence of irrelevant alternatives}: the
      ratio of attack rates between two organizations depends only on
      their own $(V, E)$ pairs, not on the rest of the population.
    \item[(A3)] \textit{Homogeneity}: the attack rate is homogeneous
      of positive degree in $(V, E)$ jointly --- the model does not
      require constant returns to scale, and the degree $a+b$ is a
      structural parameter (baseline priors $a = 1.0$, $b = 0.5$;
      the structural results hold for any $a, b > 0$).
  \end{enumerate}
  The multiplicative form $\lambda_A = \lambda_0 V^a E^b$ is imposed
  by the intersection principle above, under the approximate
  independence of reachability and cryptographic vulnerability stated
  in Hypothesis~\ref{hyp:independence}.
  The contest-success framework of \citet{skaperdas1996} further
  supports the proportional contest form once the attack-production
  effort $\lambda_A$ is specified: under (A1)--(A3), the admissible
  class of contest success functions (CSF) over $\lambda_A$ is the power
  ratio form $f(\lambda_A)/\sum_j f(\lambda_{A,j})$ with
  $f(\lambda_A) = \alpha\lambda_A^m$.
  The axioms (A1)--(A3) do not, by themselves, derive the internal
  $(V,E)$ decomposition of the effort variable.
  In the CSF framework, $\lambda_A = \lambda_0 V^a E^b$ plays the role
  of the effort variable $y_i$, decomposed into its cryptographic
  ($V$) and operational ($E$) components.
  A direct consequence is a strictly positive cross-partial
  $\partial^2\lambda_A/\partial V\,\partial E > 0$: vulnerability
  and exposure are complements in attack production, not
  substitutes.%
  \footnote{Baseline values $a = 1.0$ (neutral sensitivity, canonical
  lottery case~\citep{tullock1980,skaperdas1996}) and $b = 0.5$
  (hypothesized saturation in adversarial accessibility,
  cf.~\citep{gordon2002}) serve as structural priors; the
  characterization of Sections~\ref{sec:main}--\ref{sec:consequences}
  holds for any $a, b > 0$.}
\end{hypothesis}

\begin{remark}[Structural analogy with non-extensive statistical mechanics]
\label{rem:tsallis}
  The contest outcome in Hypothesis~\ref{hyp:poisson} admits a
  structural analogy with non-extensive statistical mechanics.
  Writing $u = V^a E^b$, the probability that the defender wins is
  \begin{equation}
    P_D \;=\; \frac{\theta}{u + \theta}
      \;=\; e_{q=2}^{-u/\theta},
    \label{eq:tsallis}
  \end{equation}
  where $e_q^{-x} = [1+(q-1)x]^{-1/(q-1)}$ is the Tsallis
  $q$-exponential~\citep{tsallis1988}.
  The value $q = 2$ is not a fitted parameter; it is determined by
  the binary contest structure of Hypothesis~\ref{hyp:composition}
  and corresponds to $d = -1$ in the $d$-Arrhenius
  parameterization~\citep{aquilanti2010,aquilanti2017}.
  The connection is suggestive rather than constitutive:
  the IEQ does not derive from non-extensive thermodynamics, and the
  analogy carries no physical interpretation beyond the contest structure
  of Hypotheses~\ref{hyp:poisson}--\ref{hyp:composition}.
  It situates the defender-win attenuation of Eq.~\eqref{eq:tsallis}
  within a known class of $q$-exponential
  systems~\citep{vallianatos2009,rufino2022};
  formal unification is left for subsequent work.
  The epistemological point, however, is precise: the Tsallis form at
  $q = 2$ is not imported from non-extensive thermodynamics as a
  primitive --- it is consistent with the operational contest structure
  (A1)--(A3), which imposes admissibility conditions solely on
  observable outcomes.
  Non-standard statistical forms that emerge from observable-level
  admissibility constraints, rather than from foundational postulates
  about the underlying system, arise in structurally analogous ways
  in formal frameworks far removed from
  cybersecurity~\citep{oliveira2026}.
\end{remark}

\begin{hypothesis}[Asymptotic independence]
\label{hyp:independence}
  Cryptographic architecture ($V$) and external attack-surface
  accessibility ($E$) are driven by structurally distinct processes.
  We treat them as approximately independent in the cross-sectional
  distribution.
\end{hypothesis}

Independence can fail within tightly integrated supply chains,
where an organization's internal cryptographic posture correlates
with its external exposure~\citep{kunreuther2003}.
Section~\ref{sec:limitations} returns to this point.

\begin{hypothesis}[Proportional hazards composition]
\label{hyp:composition}
  Given that the CRQC has arrived, the probability that the attacker
  exploits the harvested data before the defender renders it
  strategically valueless is
  \begin{equation}
    P(\text{exploit} \mid \text{CRQC}) = \frac{\lambda_A}{\lambda_A + \lambda_D}.
    \label{eq:csf}
  \end{equation}
\end{hypothesis}

Equation~\eqref{eq:csf} is the Tullock Contest Success
Function~\citep{tullock1980}: the attacker wins in proportion to their
relative rate.
It is structurally analogous to the Gordon-Loeb~\citep{gordon2002}
framework, in which the marginal benefit of security investment equals
its marginal cost --- here operationalized as the ratio of attack and
defense rates rather than as an investment optimization.

\section{Main Result}
\label{sec:main}
\begin{figure*}[!t]
  \centering
  \includegraphics[width=\textwidth]{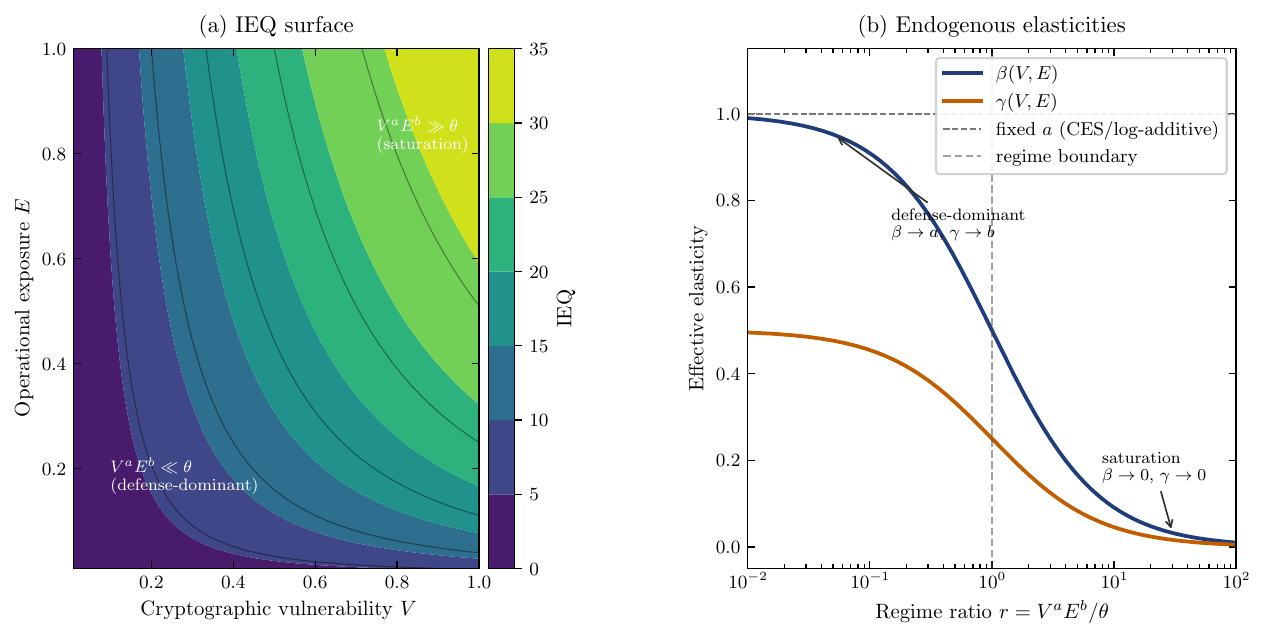}
  \caption{Structural properties of the IEQ under
    Hypotheses~1--3 with $a=1.0$, $b=0.5$, $\theta=1.0$,
    $H=0.6$, $M=1.15$.
    \textbf{(a)}~IEQ surface over the $(V,E)$ unit square
    (Eq.~\eqref{eq:IEQ}).
    Contour density reflects score sensitivity: high in the
    defense-dominant regime ($V^aE^b \ll \theta$, lower left)
    and low near saturation ($V^aE^b \gg \theta$, upper right).
    The observed range reflects the representative parameter
    $\theta=1.0$; the full $[0,100]$ scale is realized at lower
    $\theta$ values.
    \textbf{(b)}~Endogenous elasticities $\beta(V,E)$ and
    $\gamma(V,E)$ (Eq.~\eqref{eq:elasticities}) as functions of
    the regime ratio $r = V^aE^b/\theta$.
    The dashed horizontal line shows the fixed elasticity assumed
    by CES and log-additive specifications; the multiplicative
    family predicts elasticities that decrease continuously from
    their prior values ($r \ll 1$) to zero ($r \gg 1$),
    consistent with regime duality (Theorem~\ref{thm:main}).}
  \label{fig:structural}
\end{figure*}

Before stating the theorem, we define the temporal component of exposure.

\begin{definition}[Temporal hazard $H$]
  Let $F_q(t, \mu_s)$ be the logistic CDF of $T_Q$ parametrized by
  sector-specific median maturity year $\mu_s$~\citep{mosca2018} and
  slope $k = \ln(9)/10$ (chosen so that $F(\mu_s + 10) \approx 0.9$
  and $F(\mu_s - 10) \approx 0.1$, placing approximately $80\%$ of
  probability mass within a 20-year interval centred on the sector median).
  The \emph{temporal hazard} is
  \begin{equation}
    H \;=\; F_q(t_0 + T_D,\, \mu_s)
      \;=\; \frac{1}{1 + e^{-k(t_0 + T_D - \mu_s)}}.
    \label{eq:H}
  \end{equation}
  $H$ is the marginal probability that the CRQC arrives within the
  data's adversarial horizon.
  It depends on the time horizon $(t_0, T_D, \mu_s)$ but not on
  the organization's cryptographic posture $(V, E)$.
\end{definition}

\begin{theorem}[Model-implied form of an HNDL exposure score]
\label{thm:main}
  Under Hypotheses~\ref{hyp:poisson}--\ref{hyp:composition},
  the HNDL compromise probability factors as
  \begin{equation}
    P_{\mathrm{HNDL}} \;=\; H \cdot \frac{V^a E^b}{V^a E^b + \theta},
    \qquad \theta = \frac{\mu}{\lambda_0}.
    \label{eq:P}
  \end{equation}
  In a neighborhood of any operating point $(V_0, E_0)$, the
  log-linear approximation gives
  \begin{equation}
    \ln P_{\mathrm{HNDL}} \;\approx\; \ln H + \beta\,\ln V + \gamma\,\ln E + c,
    \label{eq:loglinear}
  \end{equation}
  where the \emph{elasticities} $\beta$ and $\gamma$ are endogenous:
  \begin{equation}
    \beta(V,E) = \frac{a\,\theta}{V^a E^b + \theta},
    \qquad
    \gamma(V,E) = \frac{b\,\theta}{V^a E^b + \theta}.
    \label{eq:elasticities}
  \end{equation}
  An exposure index consistent with
  Hypotheses~\ref{hyp:poisson}--\ref{hyp:composition} is therefore
  \begin{equation}
      \mathrm{IEQ}
      \;=\; 100 \cdot \min\!\bigl(1,\; H' \cdot V'^{\,\beta(V,E)}
        \cdot E'^{\,\gamma(V,E)} \cdot M\bigr),
      \label{eq:IEQ}
  \end{equation}
  where primes denote $\max(\cdot,\varepsilon)$ floors and $M \geq 1$
  is a governance penalty multiplier encoding qualitative risk factors
  (e.g., regulatory non-compliance, absence of cryptographic inventory)
  not captured by the observable signals $V$ and $E$.
  The index $\mathrm{IEQ}$ is an operational prioritisation index
  derived from a local log-linear approximation to
  $P_{\mathrm{HNDL}}$.
  Within fixed $(H,\theta,M)$ and floor regimes it preserves the
  intended structural ordering; it should not be interpreted as a
  globally calibrated probability or as a global scalar transform of
  $P_{\mathrm{HNDL}}$ across heterogeneous sectors and governance
  states.
\end{theorem}

\begin{remark}[Two levels of the model]
\label{rem:twolevel}
  The model operates at two distinct levels.
  The \emph{model-implied HNDL compromise probability} is:
  \[
    P_{\mathrm{HNDL}} = H \cdot \frac{V^a E^b}{V^a E^b + \theta}.
  \]
  The \emph{operational exposure index} is:
  \[
    \mathrm{IEQ} = 100 \cdot
      \min\!\bigl(1,\;\varphi(P_{\mathrm{HNDL}},M)\bigr),
  \]
  where $\varphi$ encodes three implementation-specific
  transformations: (i)~$\max(\cdot,\varepsilon)$ floors on $H$, $V$,
  $E$ to prevent score collapse at boundary values;
  (ii)~local log-linearisation of $P_{\mathrm{HNDL}}$ via the
  endogenous elasticities $\beta(V,E)$, $\gamma(V,E)$ of
  Eq.~\eqref{eq:elasticities}, providing a computationally tractable
  per-organisation approximation; and (iii)~a governance penalty
  multiplier $M \geq 1$ encoding qualitative risk factors not captured
  by observable signals.
  The IEQ should not be interpreted as a calibrated probability once
  these transformations are applied.
  It is an operational prioritisation index whose qualitative ordering
  is locally consistent with $P_{\mathrm{HNDL}}$ within fixed
  $(H,\theta,M)$ and floor regimes; across heterogeneous sectors and
  governance states it does not constitute a global scalar transform
  of $P_{\mathrm{HNDL}}$.
  The structural analysis of
  Sections~\ref{sec:main}--\ref{sec:consequences} concerns
  $P_{\mathrm{HNDL}}$ directly; IEQ approximates its structural
  properties via the local log-linearisation of
  Eq.~\eqref{eq:elasticities}.
\end{remark}

\begin{proof}
  By Hypothesis~\ref{hyp:poisson}, exploitation and remediation are
  exponential races with rates $\lambda_A$ and $\lambda_D = \mu$
  from Eq.~\eqref{eq:rates}.
  Hypothesis~\ref{hyp:composition}, Eq.~\eqref{eq:csf}, gives
  \[
    P(\text{exploit} \mid \text{CRQC})
    = \frac{\lambda_A}{\lambda_A + \lambda_D}
    = \frac{V^a E^b}{V^a E^b + \theta},
  \]
  after dividing through by $\lambda_0$ and setting
  $\theta = \mu/\lambda_0$.
  We additionally assume $T_Q \perp (V, E)$: the arrival of a
  cryptographically relevant quantum computer is a global technological
  process exogenous to any single organization's cryptographic posture,
  consistent with the horizon analysis of \citet{mosca2018} and
  \citet{joseph2022}. Under this assumption:
  \begin{align*}
    P_{\mathrm{HNDL}}
    &= P(T_Q < t_0+T_D) \cdot P(\text{exploit} \mid T_Q < t_0+T_D) \\
    &= H \cdot \frac{V^a E^b}{V^a E^b + \theta},
  \end{align*}
  which is Eq.~\eqref{eq:P}.
  Differentiating the logarithm of Eq.~\eqref{eq:P}:
  \begin{align*}
    \frac{\partial \ln P}{\partial \ln V} &= \frac{a\,\theta}{V^a E^b + \theta}
    \equiv \beta(V,E), \\
    \frac{\partial \ln P}{\partial \ln E} &= \frac{b\,\theta}{V^a E^b + \theta}
    \equiv \gamma(V,E).
  \end{align*}
  These are the elasticities in Eq.~\eqref{eq:elasticities}.
  A first-order Taylor expansion of $\ln P$ in $(\ln V, \ln E)$ around
  $(V_0, E_0)$ recovers the log-linear approximation~\eqref{eq:loglinear}.
  Scaling by 100 and applying the governance multiplier $M$ and floor
  clips yield Eq.~\eqref{eq:IEQ}.
\end{proof}

The factorization $P_{\mathrm{HNDL}} = H \cdot P(\text{exploit} \mid \text{CRQC})$
follows from the law of total probability: $P(A \cap B) = P(A) \cdot P(B \mid A)$,
where $A = \{T_Q < t_0 + T_D\}$ is the event that the CRQC arrives within
the adversarial horizon and $B$ is the event that exploitation precedes
remediation. The exogeneity assumption $T_Q \perp (V,E)$ of
Hypothesis~\ref{hyp:independence} ensures that $P(B \mid A)$ depends only
on $\lambda_A$ and $\lambda_D$.

The Cobb-Douglas form in Eq.~\eqref{eq:loglinear} is not an estimated
production function; it is a local approximation derived directly from
the proportional hazards structure of Hypothesis~\ref{hyp:composition}.
The elasticities in Eq.~\eqref{eq:elasticities} vary continuously with
the realized $(V, E)$ pair, distinguishing this framework from the
standard econometric log-linear approximation and its critiques~\citep{shaikh1974}.

Figure~\ref{fig:structural} illustrates both consequences of
Eq.~\eqref{eq:P}.
Panel~(a) shows the IEQ surface from Eq.~\eqref{eq:IEQ}: the score is
most sensitive to changes in $V$ and $E$ in the defense-dominant region
(lower left, where $V^aE^b \ll \theta$) and saturates in the
attack-dominant region (upper right, where $V^aE^b \gg \theta$).
Panel~(b) shows $\beta(V,E)$ and $\gamma(V,E)$ from
Eq.~\eqref{eq:elasticities} as functions of the regime ratio
$r = V^a E^b / \theta$.
When $r \ll 1$ the elasticities equal their prior values $a$ and $b$;
when $r \gg 1$ they approach zero.
The dashed line in panel~(b) marks the constant elasticity assumed by
CES and log-additive specifications --- a special case that
Eq.~\eqref{eq:elasticities} reaches only in the limit $\theta \to \infty$.

Under (A1)--(A3), the multiplicative power family is the canonical
specification for the aggregate effort $\lambda_A$.
CES and log-additive forms for the internal $V$--$E$ aggregation
are excluded by the intersection principle
(Hypothesis~\ref{hyp:poisson}): a cryptographic compromise requires
the \emph{simultaneous} presence of $V$ and $E$, and CES allows
unlimited substitution between them.
\citet{skaperdas1996} (Theorem~2) further constrains the contest
success function over $\lambda_A$ via axioms (A1)--(A3), but does not
by itself determine the internal $V$--$E$ decomposition.
Threshold forms do not satisfy (A3): they introduce superlinear returns at
the activation boundary.

Two structural properties of Eq.~\eqref{eq:P} are worth naming.

\textit{Regime duality.}
When $V^a E^b \ll \theta$ (defense dominates), the elasticities in
Eq.~\eqref{eq:elasticities} approach $\beta \to a$ and $\gamma \to b$:
the score is maximally sensitive to both dimensions.
When $V^a E^b \gg \theta$ (attack dominates), both elasticities approach
zero and the score saturates --- reducing $V$ or $E$ marginally yields
diminishing returns.
Figure~\ref{fig:structural}(b) makes this transition visible.

\textit{Migration efficiency.}
Organizations with higher current exposure face greater marginal risk
from delayed migration.
From Eq.~\eqref{eq:P}, the marginal sensitivity of $P_{\mathrm{HNDL}}$
to $V$ and $E$ is highest when $V^aE^b \ll \theta$ (defense dominates)
and approaches zero at saturation, a direct consequence of the contest
structure in Eq.~\eqref{eq:csf}: migrating early, while the organization
is in the high-sensitivity regime, yields the greatest marginal return.

\section{Consequences}
\label{sec:consequences}

\begin{corollary}[Loss of interaction structure under additive separability]
\label{cor:degenerate}
  Any purely additive separable exposure score of the form
  $S = \sum_i w_i x_i$ (additive composite) or
  $S = f(\mathrm{rank}(\mathbf{x}))$ (ordinal) uses fixed marginal
  contributions that are independent of the organization's position in
  the $(V,E)$ plane, eliminating the $V \times E$ interaction and
  treating vulnerability and accessibility as substitutes rather than
  complements in attack production.
  No such score preserves the interaction structure induced by
  Hypotheses~1--3 within a proportional-hazards HNDL framework.
\end{corollary}

\begin{proof}
  For any additive $S$, the cross-partial in natural coordinates
  $\partial^2 S / \partial V \partial E = 0$: the marginal contribution
  of $V$ is independent of $E$ by construction.
  For $P_{\mathrm{HNDL}}$, direct computation from Eq.~\eqref{eq:P}
  gives
  \begin{equation*}
    \frac{\partial^2 P_{\mathrm{HNDL}}}{\partial V\,\partial E}
    = \frac{H\,ab\,\theta\,V^{a-1}E^{b-1}\,\bigl(\theta - V^aE^b\bigr)}
           {\bigl(V^aE^b + \theta\bigr)^3},
  \end{equation*}
  which is nonzero wherever $V^aE^b \neq \theta$.
  Therefore no additive score can simultaneously match $\partial P/\partial V$
  and $\partial P/\partial E$ across the full $(V,E)$ domain.
  The log-log cross-partial,
  \begin{equation}
    \frac{\partial^2 \ln P}{\partial \ln V \,\partial \ln E}
    = -\frac{ab\theta\, V^a E^b}{(V^a E^b + \theta)^2},
    \label{eq:crosspartial}
  \end{equation}
  is strictly negative for all finite positive $a, b, V, E, \theta$:
  the endogenous elasticities $\beta(V,E)$ and $\gamma(V,E)$ depend on
  the organization's position in the $(V,E)$ plane and cannot be
  reproduced by the fixed weights of any additive composite.
  An ordinal transformation preserves the sign of first derivatives
  but cannot recover the interaction structure of
  Eq.~\eqref{eq:crosspartial}.
  In the composite-indicator literature, this is called
  \emph{full compensability}~\citep{oecd2008}: a deficit in one dimension
  is perfectly offset by a surplus in another.
  Under the HNDL model, $V$ and $E$ are complements in attack production ---
  full compensability conflicts with Hypothesis~\ref{hyp:poisson}.
\end{proof}

The opposing signs of the two cross-partials are not inconsistent,
but reflect the distinction between the attack-production technology
and the resulting compromise probability. The positive cross-partial
$\partial^2\lambda_A/\partial V\,\partial E > 0$ confirms structural
complementarity in attack generation, consistent with the
multiplicative contest structure of \citet{tullock1980} and
\citet{skaperdas1996}. The negative cross-partial
$\partial^2\ln P_{\mathrm{HNDL}}/\partial\ln V\,\partial\ln E < 0$
emerges from the Tullock denominator as a saturation effect: as $V$
and $E$ increase jointly, $P_{\mathrm{HNDL}}$ approaches its upper
bound $H$, reducing the marginal sensitivity of the score in a manner
analogous to the diminishing returns to security investment identified
by \citet{gordon2002}. The log-probability is
therefore the correct analytic object for evaluating the exposure
interactions in Corollary~\ref{cor:degenerate}.

The practical consequence is direct: a framework that scores
cryptographic vulnerability independently of data exposure will
produce both false positives and false negatives.
The error is structural, not numerical.

\section{Structural Specification Diagnostics}
\label{sec:empirical}

The operational implementation maps heterogeneous observable and
declarative signals into the structural variables $V$ and $E$.
Those implementation details are product-specific and commercially
sensitive; the present paper concerns the structural form of the
exposure model rather than the full measurement pipeline.

The framework was instantiated over $N \approx 40{,}000$ organizations
using automated external observation as the source of $(V, E)$ signals.
Signal influence was measured via the Sobol total-effect
index~\citep{campolongo2000,saisana2005,becker2017,nguyen2025}, which accounts for interaction
effects under the multiplicative structure of Eq.~\eqref{eq:P}.
Uncertainty was assessed by Monte Carlo perturbation.

A penalized spline~\citep{hastie2009} fitted to all input signals (Delta Layer) found no
stable residual structure beyond the structural model, consistent with
Eq.~\eqref{eq:P} capturing the dominant variance.

Non-nested model comparison via the Vuong statistic~\citep{vuong1989}
was applied against CES, log-additive, and threshold alternatives as
an internal specification check: the comparison tests whether the
scoring pipeline is self-consistent with the structural form of
Eq.~\eqref{eq:P}, not whether the score predicts ground-truth HNDL
outcomes (which are unobservable in the pre-CRQC regime).
The observed log-cross-partial
$\partial^2\ln\hat{P}_{\mathrm{HNDL}}/\partial\ln V\,\partial\ln E < 0$
is consistent with the Tullock-denominator saturation structure of
Eq.~\eqref{eq:crosspartial}; this is a saturation effect in the
compromise probability, not a sign reversal of the positive
cross-partial of the attack-generation technology $\lambda_A$
(Section~\ref{sec:consequences}).
The CES form yielded a positive log-cross-partial, inconsistent with
Eq.~\eqref{eq:crosspartial}, confirming the structural
distinction identified in Corollary~\ref{cor:degenerate}.%
\footnote{Vuong test results and cross-partial estimates are reported
in full in the Supplementary Technical Material (Section~11).}

Operational calibration details are outside the scope of this
structural analysis.
This framing follows the epistemological standard of
\citet{gordon2002}, who acknowledge that empirical determination
of threat probabilities and potential loss is outside the scope of
their structural model, and of the OECD composite-indicator
framework~\citep{oecd2008}, which distinguishes internal structural
consistency from causal identification.

\section{Limitations}
\label{sec:limitations}

Four limitations are worth stating plainly.

\textit{Poisson approximation.}
Hypothesis~\ref{hyp:poisson} models attack and defense as
constant-rate exponential processes.
Correlated attack campaigns violate the independence of increments.

\textit{Independence of $V$ and $E$.}
Hypothesis~\ref{hyp:independence} holds asymptotically in large
heterogeneous populations but fails for tightly integrated sectors.
Supply-chain interdependencies~\citep{kunreuther2003} can correlate
internal cryptographic posture with external exposure.
The empirical Spearman $\rho(V,E) = 0.078$ ($p = 0.001$) supports
approximate independence at the population level; localized signal
pairs can exhibit substantially higher correlation.

\textit{Adversarial shelf life $T_D$.}
This is a sector-level prior, not an auditable metric.
\citet{mosca2018} treats it as a strategic horizon estimate;
\citet{joseph2022} document that assets such as trade
secrets and medical records carry multi-decade shelf lives, but the
uncertainty is inherent.

\textit{No ground truth.}
There is no dataset of confirmed HNDL exploitations.
Absolute calibration of Eq.~\eqref{eq:P} is infeasible in the
pre-CRQC regime.
Signal influence is therefore assessed against the internal variance
structure of observable signals rather than against confirmed outcomes.
This target-free evaluation is methodologically consistent with the
sensitivity-based weighting framework of \citet{nguyen2025}, in which
weights are derived entirely from the distribution of inputs rather than
from labelled outcomes or an external target variable.
The contribution is structural: the form of the score, not specific
calibrated constants~\citep{gordon2002}.

\section{Conclusion}
\label{sec:conclusion}

The HNDL threat has a precise structure.
An adversary needs two things at once: a quantum-vulnerable cipher ($V$)
and access to the traffic ($E$).
The two are complements, not substitutes.

Three axioms about how attack rates are produced imply the exposure
probability takes the form of Eq.~\eqref{eq:P}
(Theorem~\ref{thm:main}).
The elasticities in Eq.~\eqref{eq:elasticities} are endogenous:
they depend on where the organization sits in the vulnerability-exposure
plane, not on global constants.
Figure~\ref{fig:structural} shows what this looks like across the full
$(V,E)$ operating range.

Two consequences follow directly.
First, additive and ordinal scores cannot reproduce the $V \times E$
interaction structure (Corollary~\ref{cor:degenerate}, Eq.~\eqref{eq:crosspartial}).
Recalibration alone cannot restore the missing interaction term; it is absent by construction in any separable form.
Second, migration efficiency is highest at low current exposure
(regime duality of Eq.~\eqref{eq:elasticities}): start early, while
the organization is in the high-sensitivity regime.

An internal specification diagnostic over $\approx 40{,}000$ organizations
confirmed that the scoring pipeline is self-consistent with the
structural form: the observed log-cross-partial had the sign predicted
by Eq.~\eqref{eq:crosspartial}, and the CES alternative yielded the opposite sign.

Natural extensions include network-contagion models relaxing axiom~(A2)
for supply-chain-integrated sectors, and empirical estimation of
$\theta = \mu/\lambda_0$ from observed PQC migration rates.

\section*{CRediT authorship contribution statement}
\textbf{Matheus Rufino:} Conceptualization, Methodology, Writing -- original draft.
\textbf{Rafael Duarte Marcelino:} Software, Investigation, Formal analysis,
Data curation, Validation.
\textbf{Julio Smanioto Garcia:} Visualization, Writing -- review \& editing.

\section*{Declaration of competing interests}
The authors are co-founders of GWK Security, which develops cybersecurity
measurement products based on the IEQ methodology presented in this paper.
The formal model, axioms, and proofs reported here are independent of
any customer-identifying or proprietary implementation data.
The authors declare no other competing financial interests or personal
relationships that could have appeared to influence the work reported
in this paper.

\section*{Acknowledgements}
This research did not receive any specific grant from funding agencies
in the public, commercial, or not-for-profit sectors.

\section*{Declaration of generative AI and AI-assisted technologies
in the manuscript preparation process}
During the preparation of this work the authors used GitHub Copilot
(GitHub, Inc.) for \LaTeX{} formatting and code-editing assistance.
No AI system generated, accessed, or interpreted proprietary
operational data.
The authors reviewed and edited all content as needed and take full
responsibility for the content of the published article.

\section*{Data availability}
The operational signals used in the specification diagnostic
(Section~\ref{sec:empirical}) are derived from automated external
observation of live organizations and contain commercially sensitive
information that cannot be made publicly available.
No customer-identifying, organization-level, domain-level, or
signal-level data are reported in this paper.
The formal model, axioms, proofs, and all equations are fully
self-contained and reproducible from the material presented in
this paper.

\bibliographystyle{cas-model2-names}
\bibliography{ieq_paper}

\end{document}